\documentclass[sigconf]{acmart}

\usepackage{latexsym,amsmath,amsfonts,amscd,verbatim,array,mathdots}
\usepackage{algorithm}

\newtheorem{theorem}{Theorem}

\DeclareMathOperator{\Q}{\mathbb{Q}}
\DeclareMathOperator{\R}{\mathbb{R}}

\DeclareMathOperator{\Hi}{\mathscr{H}}
\DeclareMathOperator{\Z}{\mathbb{Z}}
\DeclareMathOperator{\hht}{ht}
\DeclareMathOperator{\csosone}{\mathtt{csos1}}
\DeclareMathOperator{\csostwo}{\mathtt{csos2}}
\DeclareMathOperator{\csosthree}{\mathtt{csos3}}
\DeclareMathOperator{\sep}{sep}
\DeclareMathOperator{\Disc}{Disc}
\usepackage{xcolor}

\usepackage{myalgo}

\DeclareMathOperator{\trace}{tr}
\DeclareMathOperator{\N}{\mathbb{N}}

\newcommand{\sdpfun}[3]{\texttt{sdp}(#1,#2,#3)}
\newcommand{\sdp}{\texttt{sdp}}
\newcommand{\nsdp}{n_\text{sdp}}
\newcommand{\msdp}{m_\text{sdp}}
\newcommand{\roundfun}[2]{\texttt{round}(#1,#2)}

\newcommand{\ldlfun}[1]{\texttt{ldl}(#1)}

\newcommand{\cholesky}{\texttt{cholesky}}
\newcommand{\choleskyfun}[3]{\texttt{\cholesky}(#1,#2,#3)}

\newcommand{\hasrealrootoncircle}{\texttt{hasrealrootoncircle}}

\newcommand{\rootsfun}[2]{\texttt{complexroots}(#1,#2)}

\newcommand{\coeffs}{\texttt{coeffs}}

\def\ucircle{\mathscr{C}}
\def\C{\mathbb{C}}
\copyrightyear{2022}
\acmYear{2022}
\setcopyright{acmcopyright}\acmConference[ISSAC '22]{Proceedings of the 2022 International Symposium on Symbolic and Algebraic Computation}{July 4--7, 2022}{Villeneuve-d'Ascq, France}
\acmBooktitle{Proceedings of the 2022 International Symposium on Symbolic and Algebraic Computation (ISSAC '22), July 4--7, 2022, Villeneuve-d'Ascq, France}
\acmPrice{15.00}
\acmDOI{10.1145/3476446.3535480}
\acmISBN{978-1-4503-8688-3/22/07}
\begin{document}

\title[Exact SOHS Decompositions of Trigonometric Univariate  Polynomials]{
Exact SOHS Decompositions of Trigonometric Univariate Polynomials with Gaussian Coefficients}

\author{Victor Magron}
\affiliation{%
  \institution{CNRS, LAAS}
  \streetaddress{7 avenue du Colonel Roche}
  \city{Toulouse}
  \country{France}
  \postcode{F-31400}
}

\author{Mohab Safey El Din}
\affiliation{%
  \institution{Sorbonne Universit\'e, CNRS, LIP6}
  \streetaddress{F-75005, Paris, France}
  \city{Paris}
  \country{France}}

\author{Markus Schweighofer}
\affiliation{%
  \institution{University of Konstanz}
  \city{Konstanz}
  \country{Germany}
}

\author{Trung Hieu Vu}
\affiliation{%
    \institution{Sorbonne Universit\'e, CNRS, LIP6}
    \streetaddress{F-75005, Paris, France}
    \city{Paris}
    \country{France}}

\begin{abstract}
%
%
  Certifying the positivity of trigonometric polynomials is of first importance for
  design problems in discrete-time signal processing.
  It is well known from the Riesz-Fej\'{e}r spectral factorization theorem that
  any trigonometric univariate  polynomial non-negative on the unit circle can be decomposed as a
  Hermitian square with complex coefficients. Here we focus on the case of
  polynomials with Gaussian integer coefficients, i.e., with real and imaginary
  parts being integers.

  We design, analyze and compare, theoretically and practically, three hybrid
  numeric-symbolic algorithms computing weighted sums of Hermitian squares
  decompositions for trigonometric univariate polynomials positive on the unit circle with
  Gaussian coefficients. The numerical steps the first and second algorithm rely
  on are complex root isolation and semidefinite programming, respectively. An
  exact sum of Hermitian squares decomposition is obtained thanks to compensation techniques. The
  third algorithm, also based on complex semidefinite programming, is an
  adaptation of the rounding and projection algorithm by Peyrl and Parrilo.

  For all three algorithms, we prove bit complexity and output size estimates
  that are polynomial in the degree of the input and linear in the maximum
  bitsize of its coefficients.
  We compare their performance on randomly chosen
   benchmarks, and further design
  a certified finite impulse filter.
\end{abstract}
\thanks{%
  The authors are supported by the European Union's Horizon 2020 research and
  innovation programme under the Marie Sk\l{}odowska-Curie grant agreement N.
  813211 (POEMA)}


\maketitle

\section{Introduction}

In this paper, we denote by (resp. $\N$) $\Z$ the set of (resp. non-negative)
integers and by $\Q$ (resp. $\Q_+$), $\R$ and $\C$ the fields of rational (resp. positive rational), real and complex
numbers.
For a complex variable or number $f$, we denote by $\bar{f}$ the
associated conjugate variable or number.

Let $\Hi[z]$ be the set of trigonometric univariate polynomials defined as a subset of Laurent polynomials with complex coefficients and
complex variable $z$ as follows:
%
\begin{equation}\label{Cfz}
f(z) =  f_0+\left(\frac{f_1}{z}+\bar f_1z\right) +\dots+\left(
    \frac{f_d}{z^d}+\bar f_dz^d\right),
\end{equation}
with $f_0 \in \R$ and $d \in \N$.
By convention, when $f_d\neq 0$, $d$ is the degree of $f$; the degree of the
zero polynomial is $-\infty$.
%

In this paper, since we work with base fields of characteristic zero, we see
more polynomials through the evaluation maps they define than as algebraic
objects. Note that for $f \in \Hi[z]$, the restriction of the map $\zeta \mapsto
f(\zeta)$ over the {\em unit circle} $\ucircle := \{\zeta\in\C:|\zeta|=1\}$
coincides with the evaluation map defined by the polynomial
\begin{equation*}
g(z) = f_0  +\left({f_1}{\bar z}+\bar f_1z\right) +\dots+\left(
 {f_d}{\bar z}^d+\bar f_dz^d\right),
\end{equation*}
since $\bar \zeta = \zeta^{-1}$ for $\zeta \in \ucircle$. 
Note also that for any $\zeta\in \ucircle$, $g(\zeta)\in \R$ so that $g$ is a
{\em Hermitian} polynomial.
Finally, note that for any
Hermitian polynomial $g$, there exists $f \in \Hi[z]$ such that the restrictions
to $\ucircle$ of the maps $\zeta \mapsto g(\zeta)$ and $\zeta \mapsto f(\zeta)$
coincide (this is due to the fact that all points in $\ucircle$ satsify $ \zeta \bar \zeta
= 1$).



For $h(z) = \sum_{k=0}^d h_k z^k\in \C[z]$, we define $h^\star(z) = \sum_{k=0}^d
\bar h_k  z^{-k}$.



One says that $f$ is a \textit{sum of Hermitian squares} (SOHS) if there exist
some $r \in \N-\{0\}$ and polynomials $s_1,\dots,s_r $ in $\C[z]$ such that
$$f(z)=\sum_{j=1}^{r} s_j (z) s_j^\star ({z}) \,.$$
This terminology of Hermitian squares comes from the above discussion as $s_j^\star( \zeta) = s_j (\bar \zeta)$ for all $\zeta \in \ucircle$.
By the Riesz-Fej\'{e}r spectral factorization theorem (see, e.g., \cite[Theorem
1.1]{dumi2017}), any trigonometric univariate polynomial which is non-negative over the unit
circle $\ucircle$ can be written as an Hermitian square, i.e., an SOHS with a single term.

Its proof \cite[pp. 3--5]{dumi2017} is constructive but requires to manipulate
exactly all $2d$ complex roots of the associated Laurent polynomial. Usually,
this algorithm is applied with approximate computations, leading to approximate
certificates of non-negativity over $\ucircle$.

The subset of $\Hi[z]$ with \textit{Gaussian integers} coefficients, i.e., all
coefficients $f_i$ lie in $\Z + i \Z$ (where $i$ is the standard imaginary unit)
is denoted by $\Hi(\Z)[z]$. In this paper, we focus on the computation of {\em
  exact} certificates of non-negativity of polynomials in $\Hi(\Z)[z]$ by means
of {\em exact} SOHS decompositions.


Our motivation comes from design problems in discrete-time signal processing. In
particular, for the design of finite impulse response (FIR) filters in signal
processing, minimizing the stopband energy is a crucial issue \cite[Chapter
5]{dumi2017}. Computing exact SOHS decompositions of trigonometric univariate polynomials
in this context appears to be a natural computational issue.

\paragraph*{Our contribution}
We design three exact algorithms to compute SOHS decompositions of polynomials
in $\Hi(\Z)[z]$ that are \textit{positive} over the unit circle $\ucircle$.
These algorithms are based on perturbation-compensation or rounding-projection
techniques. We analyze their bit complexities and output size as well.

We use the height of a polynomial with rational coefficients to measure its
\textit{bitsize} that is defined as follows. The bitsize of an integer $b$ is
denoted by $\hht(b):=\lfloor\log_2(|b|)\rfloor+1$ with $\hht(0):= 1$, where
$\log_2$ is the binary logarithm. Given $a\in\Z$ and $b\in\Z$ with $b\neq 0$ and
$\gcd(a,b)=1$, we define $\hht\left ({a} / {b} \right )=\max( \hht(a),\hht(b))$. We
define the bitsize of a Gaussian rational number as $\hht(a+ib)=\max(\hht(a),
\hht(b) )$, where $a,b\in \Q$. For a non-zero polynomial $f$ with Gaussian
rational coefficients, we define $\hht(f)$ as the maximum bitsize of the
non-zero coefficients of $f$.

For two maps $p,q:\N^m\to \R$, one writes ``$p(v)=O(q(v))$" when there exists
$b\in \N$ such that $p(v)\leq b q(v)$, for all $v\in \N^m$. We use the notation
$p(v)= \widetilde O (q(v))$ when $p(v)=O(q(v)\log^kq(v))$
for some $k\in \N$.


The first algorithm we design, called $\csosone$, is a perturbation-compensation
one in which the numerical step computes an approximate SOHS decomposition for a
well-chosen perturbation of the input polynomial with complex root isolation.

\vspace*{-0.1cm}

\begin{theorem}\label{thm:bit1} Let $f\in \Hi(\Z)[z]$ be positive on $\ucircle$ of degree $d$ and coefficients of maximum bitsize $\tau$. There exists
  an Algorithm $\csosone$ which on input $f$ computes an SOHS decomposition of
  $f$ with Gaussian (or Gaussian modulus)  coefficients using at most $\widetilde{O}\left (d^6(d+\tau) \right )$
  bit operations. In addition, the maximum bitsize of the output coefficients is
  bounded from above by $\widetilde{O}(d^5(d+\tau))$.
\end{theorem}

The two other algorithms, called $\csostwo$ and $\csosthree$, are based on
complex \emph{semidefinite programming} (SDP).
SDP consists of minimizing a linear function over a set of matrices constrained to have non-negative eigenvalues; see  \cite{wolkowicz2012handbook}.
In $\csostwo$, we compute an
approximate SOHS decomposition for the perturbation by using complex SDP
solving. Algorithm $\csosthree$ is an adaptation of the rounding-projection
algorithm raised by Peyrl and Parrilo \cite{peyrl2008}. These algorithms are
more expensive compared to the first one because we replace complex root
isolation by complex SDP solving. Despite their worse complexity, they allow one
to handle constrained optimization problems and to design filters.

\begin{theorem}\label{thm:bit2}
  Let $f\in \Hi(\Z)[z]$ be positive on $\ucircle$ of degree $d$ and
  coefficients of maximum bitsize $\tau$. There exist algorithms $\csostwo$ and
  or $\csosthree$ which on input $f$ output an SOHS decomposition of $f$ with (modulus) of Gaussian coefficients
  using at most
  $\widetilde{O}(d^{13}(d+\tau)^2)$ 
  bit operations. In addition, the maximal bitsize of the output
  coefficients is bounded from above by $\widetilde{O}(d^6(\tau+d))$.
\end{theorem}

These algorithms have been implemented using the {\sc Julia} programming
language \cite{bezanson2017julia}.
We report on practical experiments that Algorithm $\csosone$ runs faster than the other algorithms; that coincides with the obtained complexity.
Furthermore, we rely on $\csosthree$ to design filters in a certified way.

\textit{Related works.} Computation of exact weighted sums of squares decompositions
of a univariate polynomial $f\in\Q[z]$ has been studied in
\cite{schweighofer1999,chevillard2011,univsos}.
Many of the techniques developed
in this paper are borrowed from these previous works.
We also mention \cite{krick2021univariate} which allows to compute certificates of
non-negativity of univariate polynomials sharing common real roots with another
univariate polynomials.
Note that computing sums of squares decompositions of
non-negative univariate polynomials with rational coefficients is easier from a complexity viewpoint, according to
the estimate in \cite[Theorem 4.4]{univsos}.


It should be noted that our SOHS decomposition problems
can be translated into sums of squares decompositions of bivariate polynomials with real
variables, which are positive on $\ucircle$.
This is a topical computational issue popularized by the original
papers of Lasserre and Parrilo \cite{lasserre2001,parrilo2000} allowing to
compute approximate sums of squares decompositions based on semi-definite
programming.
Hybrid symbolic-numeric turning these approximate certificates to
exact ones are given in
\cite{peyrl2008,kaltofen2012,guo2012,magron2021exact,magron2021sum}.
Exact certificates for several special families of polynomials have been given,
for instance SAGE/SONC polynomials \cite{magron2019,wang2020second}, polynomials
lying in the interior of the SOS cone \cite{magron2021exact}, and polynomials
whose gradient ideals are zero-dimensional and radical \cite{magron2021sum}.


Applying those results to the bivariate setting we can reduce our problems to,
\cite[Theorem 16]{magron2021exact} yields bit complexity and output bitsize
estimates which are exponential in the input degree of $f$ and in the maximum
bitsize of its coefficients.


Note that in our case, the degree of the SOHS decomposition is known in advance and we estimate the bit complexity of the coefficients involved in this decomposition.
In the multivariate case, the situation is more delicate and providing degree bounds is already a challenge.
Recent efforts have been pursued for polynomials positive over the \textit{unit sphere} \cite{reznick1995uniform,fang2021sum}, or for more general closed sets of constraints defined by finitely many polynomials \cite{baldi2021moment,mai21}.
Theorem~2 in \cite{fang2021sum} provides us a degree bound  that is linear in the number of variables
while Theorem 1.7 in \cite{baldi2021moment} and Corollary 1 in \cite{mai21} give bounds that are polynomial in the input degrees and exponential in the number of variables.

\if
To prove the above estimates on both the running time and bitsize of the coefficients in the SOHS decomposition, we introduce a lower bound for the minimum of $f$ on the unit circle:
$f-\frac{1}{2^{n}}$ is positive over the unit circle
for some $n=\widetilde{O}(d^3(d+\tau))$ (Lemma \ref{lm:epsN}).
\fi

\if{
\paragraph*{Structure of the paper}
 In the next section, we provide
  preliminary results used in the paper. Sections
 \ref{sec:csos1}--\ref{sec:csos3} introduce Algorithms $\csosone$, $\csostwo$, and $\csosthree$, respectively,
and analyze theirs bit complexities. Practical experiments are given in the last section.
}\fi
\section{Auxiliary results} \label{sec:prelim}

For a polynomial $f = f_0+\cdots+ f_dz^d \in \C[z]$ of degree $d$, the
\textit{minimal distance} between the roots $\alpha_1, \ldots, \alpha_d$ of $f$
is defined by
$$\sep(f):=\min\{|\alpha_i-\alpha_j|, \, \alpha_i\neq \alpha_j\}.$$
The \textit{norm} of $f$ is defined as $\|f\|:=|f_d|+\dots+|f_0|$.
The following lemma is an immediate consequence of the discussion following the corollary of Theorem 2 from \cite{mahler1964}.


\begin{lemma}\label{lm:sepZ}
Let $f\in \Z[i][z]$ of degree $d$ and $\tau$
be the maximum bitsize of its coefficients. 
Assume that $f$ has no multiple roots. 
The minimal distance between the roots of $f$ satisfies
    \begin{equation}\label{eq:sepZ}
        \sep(f) \geq
        \frac{\sqrt{3}}{d^{\frac{d}{2}+1}\|f\|^{d-1}}.\end{equation}
Therefore, one needs an accuracy of $\delta = \widetilde{O}(\tau d )$ to compute distinct approximations of the roots of $f$ with complex root isolation.
\end{lemma}

\begin{proof} By \cite[Theorem
  2]{mahler1964}, 
 one has:
    \begin{equation}\label{eq:sepC}
        \sep(f) \geq \frac{\sqrt{3|\Disc(f)|}}{d^{\frac{d}{2}+1}\|f\|^{d-1}},
    \end{equation}
    where $\Disc(f) = f_d^{2d-2} \prod_{j<k} (\alpha_j - \alpha_k)^2$ is the discriminant of $f$.
    Note that $\Disc(f)$ can be written as a polynomial in $f_0,\dots,f_d$ with integer coefficients, thus $\Disc(f)\in \Z[i]$ and one has $|\Disc(f)|\geq 1$ which from
    \eqref{eq:sepC}, implies
    $$\sep(f) \geq
    \frac{\sqrt{3|\Disc(f)|}}{d^{\frac{d}{2}+1}\|f\|^{d-1}} \geq
    \frac{\sqrt{3}}{d^{\frac{d}{2}+1}\|f\|^{d-1}}.$$.
\if{    
    Now we consider the two following cases.

    \textbf{Case 1:} $f$ has no multiple root in $\C$. Then, $\Disc(f)$ is
    nonzero. Since $\Disc(f)\in \Z[i]$, one has $|\Disc(f)|\geq 1$ which from
    \eqref{eq:sepC}, implies
    $$\sep(f) \geq
    \frac{\sqrt{3|\Disc(f)|}}{d^{\frac{d}{2}+1}\|f\|^{d-1}} \geq
    \frac{\sqrt{3}}{d^{\frac{d}{2}+1}\|f\|^{d-1}}.$$

    \textbf{Case 2:} $f$ has multiple roots in $\C$. Let $p$ be the square-free
    part 
      of $f$. From
\cite[Corollary 2.2
]{lang1993}, the coefficients of $p$ lie in $\Z[i]$. Clearly,
$\sep(p)=\sep(f)$ and $\|p\|\leq \|f\|$. Hence, by applying Case 1 for $p$,
where $k=\deg p\leq d$, we have
    $$\sep(f)=\sep(p) \geq
    \frac{\sqrt{3}}{k^{\frac{k}{2}+1}\|p\|^{k-1}}\geq
    \frac{\sqrt{3}}{d^{\frac{d}{2}+1}\|f\|^{d-1}}.$$
}\fi
\end{proof}

The following lemma provides a lower bound on the minimum of a real bivariate
polynomial over the unit circle in $\R^2$.

\begin{lemma}\label{lm:min}
  Let $p\in\Z[x,y]$ be a real bivariate polynomial of degree $d$ and $\tau$ be
  the maximum bitsize of its coefficients. Assume that $p$ is positive on the
  unit circle $\ucircle$. Then, the minimum of $p$ on $\ucircle$ satisfies the
  following inequality:
\begin{equation}\label{eq:min}
     p_{\min} :=   \min\{p(x,y):x^2+y^2=1\} \geq 2^{-\widetilde{O}(d^3(d+\tau))}.
\end{equation}
\end{lemma}

\begin{proof} We consider the following algebraic set:
\begin{equation*}\label{eq:m}
V := \{(x,y,m) \in \C^3 :      p(x,y) -m = \  y\frac{\partial p}{\partial x}-
        x\frac{\partial p}{\partial y}=0, x^2+y^2=1 \}.
\end{equation*}
Note that the projection of $V$ on the $m$-axis defines the critical values of
the restriction of the evaluation map $z \mapsto p(z)$ to $\ucircle$ which
contains $p_{\min}$.

Assume first that $V$ is finite. By \cite[Corollary 2]{SaSch2018}, there is a
zero-dimensional parametrization of $V$ defined by real univariate polynomials
with bitsizes upper bounded by $\widetilde{O}(d^3(d+\tau))$. Since there exists
$(x_0,y_0)$ on $\ucircle$ such that $(x_0,y_0,p_{\min})$ belongs to $V$,
$p_{\min}$ is a (non-zero) root of a univariate polynomial of degree at most
$O(d^3\tau)$. Hence, the Cauchy bound \cite{Cauchy1830} yields:
$$|p_{\min}| \geq 2^{-\widetilde{O}(d^3(d+\tau))}.$$

Assume now that $V$ is not finite. By Krull's theorem
\cite{krull1929idealtheorie}, this implies that $\ucircle$ is contained in
the complex zero set defined by $y\frac{\partial p}{\partial x}- x\frac{\partial
  p}{\partial y}=0$, whence is a factor of this polynomial. This implies that
there exists a factorization $p = p_1p_2$ where $p_1$ is a power of $x^2+y^2-c$
(where $c$ is a constant) and the zero set of the polynomial $y\frac{\partial
  p_2}{\partial x}- x\frac{\partial p_2}{\partial y}=0$ has a zero-dimensional
intersection with $\ucircle$. This yields the following analysis.
The set $V$ is the union of a 1-dimensional component containing points $(m, x,
y)$ where $(x, y)$ ranges over $\ucircle$ and $m = c-1$,  and a
0-dimensional component containing points $(m, x, y)$ which are solutions to
\[
  p(x, y)=m, \, y\frac{\partial p_2}{\partial x}- x\frac{\partial p_2}{\partial
    y}=0, \, x^2+y^2=1.
\]
Applying the first paragraph of the proof to the above system ends the proof.
\end{proof}

The upcoming result will be used to estimate bit complexities of the two
algorithms $\csosone$ and $\csostwo$.

\begin{lemma}\label{lm:epsN}
  Let $f\in \Hi(\Z)[z]$ be positive on $\ucircle$, of degree $d$ and $\tau$
  be the maximum bitsize of its coefficients. Then, there exists a positive
  integer $N=\widetilde{O}(d^3(d+\tau))$ such that $f-\frac{1}{2^{N}}$ is
  positive on $\ucircle$.
\end{lemma}
\begin{proof}

With $z = x+iy$, let us define $p(x,y) := f(x+iy)$.
Since $f \in \Hi(\Z)[z]$, one has $p \in \Z[x,y]$ with degree $2d$ and bitsize $O(\hht(d)+ \tau)$.
Clearly, $\min\{f(z):|z|=1\}=
    \min\{p(x,y):x^2+y^2=1\}=p_{\min}$.
Let us choose a positive integer $N$ such that
    $\frac{1}{2^{N}}\leq p_{\min}$.  From Lemma~\ref{lm:min},
    we conclude that $N=\widetilde{O}(d^3(d+\tau))$.
\end{proof}

The following result, stated in \cite[Lemma 2.1 \& Theorem 3.2]{bai1989floating}, will be used to investigate the bit complexity of the two algorithms $\csostwo$ and $\csosthree$ based on SDP solving.

\begin{lemma}\label{lm:Bai}
Let $Q$ be a Hermitian matrix indexed on $\{-d,\dots,d\}$,  with positive eigenvalues and rational entries.
Let $L$ be the factor of $Q$ computed by Cholesky's decomposition with
finite precision $\delta_c$. Then, $LL^T = Q+H$, where
\begin{equation}\label{eq:H}
|H_{ij}|\leq \frac{(d+2)2^{-\delta_c}\sqrt{|Q_{ii}Q_{jj}|}}{1-(d+2)2^{-\delta_c}}.
\end{equation}
In addition, if the smallest eigenvalue $\tilde{\lambda}$ of $Q$ satisfies the inequality
\begin{equation}\label{eq:2deltac}
    2^{-\delta_c} <  \frac{\tilde{\lambda}}{d^2+d+(d-1)\tilde{\lambda}},
\end{equation}
Cholesky's decomposition returns a rational nonsingular factor $L$.
\end{lemma}

\section{Algorithm based on root isolation}\label{sec:csos1}
In this section, we propose an algorithm, called $\csosone$, to compute an SOHS
decomposition of a polynomial in $\Hi(\Z)[z]$ which is positive on the unit circle
$\ucircle$. It puts into practice a perturbation-compensation procedure based on
complex roots isolation, and can be viewed as the extension of the procedure  $\texttt{univsos2}$ (stated in \cite{chevillard2011} and analyzed in \cite[\S~4]{univsos}) to the complex setting.

\subsection{Description}

\textbf{Description.} Algorithm $\csosone$ takes as input a polynomial $f\in
\Hi(\Z)[z]$ of degree $d$ which is positive on $\ucircle$. It outputs two
positive rational numbers $\varepsilon, a$, a rational number $u_0$, and two
lists of Gaussian numbers $[u_1,\dots,u_d]$ and $[\alpha_1,\dots,\alpha_d]$ such
that
\begin{multline}\label{fz}
  f(z)= \, \left(\varepsilon-2\sum_{k=1}^{d}|u_k|-u_0\right)
  +\sum_{k=1}^{d}|u_k|\left( z^k +\frac{u_k}{|u_k|}\right)
  \left( {z}^{-k} +\frac{\bar{u}_k}{|u_k|}\right) \\
  +a\prod_{k=1}^d\left( z-\alpha_k \right) \left(
    {z}^{-1}-\bar{\alpha}_k\right) \text{ with }
  \left(\varepsilon-2\sum_{k=1}^{d}|u_k|-u_0\right) > 0.
\end{multline}

\begin{algorithm}
    \caption{$\csosone$}\label{alg:csosone}
    \label{alg:sos1}
    \begin{algorithmic}[1]
      \Require $f \in \Hi(\Z)[z]$ positive on $\ucircle$ of degree $d$
      \Ensure $\varepsilon, a \in \Q_+$,
      $u_0 \in \Q$, two lists $[u_1,\dots,u_d]$ and $[\alpha_1,\dots,\alpha_d]$
      in $\Q[i]$ providing an SOHS decomposition for $f$ on $\ucircle$ as
      in \eqref{fz}:

        \State  $\delta \gets 1$, $\varepsilon := 1$ and $p\gets f(x+iy)$\label{line:subs} \Comment{$z = x+iy$}
        \While {$\hasrealrootoncircle(p-\varepsilon)$} \label{line:while0}
         $\varepsilon \gets \frac{\varepsilon}{2}$
        \EndWhile\label{line:endwhile0}
        \State boo := false
        \While {not boo}\label{line:while1}
        \State $[\alpha_1,\dots,\alpha_d] \gets
        \rootsfun{f-\varepsilon}{\delta}$\label{line:roots}
        \State $F:=\prod_{k=1}^d\left(z-
        \alpha_k\right)\left( z^{-1} -\bar{\alpha}_k\right) $\label{line:F}
        \State $a\gets\coeffs(f,0)/\coeffs(F,0)$, $u := f-\varepsilon-aF$\label{line:a}
        \State $[u_0,u_1,\dots,u_d] \gets \coeffs(u)$\label{line:lsu}
        \If {$\varepsilon> u_0 + 2\sum_{k=1}^{d}|u_k|$} boo := true\label{line:stop}
        \Else $\ \delta := 2 \delta$\label{Line:2delta}
        \EndIf
        \EndWhile\label{line:endwhile1}
        \State \Return $\varepsilon, a, u_0,
        [u_1,\dots,u_d],[\alpha_1,\dots,\alpha_d]$
    \end{algorithmic}
\end{algorithm}

In Line \ref{line:subs} we replace $z$ by $x+iy$ in $f$ where $x,y$ are (real)
variables to obtain a real bivariate polynomial $p$ of degree $2d$. Since, by
assumption, $f$ is positive over $\ucircle$, there exists $\varepsilon>0$ small
enough, such that $p - \varepsilon$ is positive on $\ucircle$. The first while loop
from Line \ref{line:while0} to \ref{line:endwhile0} computes such positive
rational number $\varepsilon$. To do so, it uses an auxiliary procedure
$\hasrealrootoncircle$, which returns true if $p-\varepsilon$ cancels on the
$\ucircle = \{(x,y)\in\R^2:x^2+y^2-1=0\}$. Such a procedure is easily obtained
with any polynomial system solver for bivariate polynomial systems. In practice
we use the real root solver {\sc msolve} \cite{msolve}.

In the second while loop from Line \ref{line:while1} to \ref{line:endwhile1},
the algorithm computes at Line \ref{line:roots} Gaussian approximations
$\alpha_1,\dots,\alpha_d$ (and their conjugates) of the complex roots of
$f-\varepsilon$ with accuracy $\delta$. This is done using a procedure
\texttt{complexroots} which on input a rational fraction and a required accuracy
$\delta$ returns all the complex roots of the numerator of the fraction at
accuracy $\delta$ (see, e.g., \cite{BiRo14}).

The idea is to obtain (up to proper scaling) an approximate SOHS decomposition
$F$ of $f - \varepsilon$.

The auxiliary $\coeffs$ procedure provides the list of coefficients of a polynomial, e.g., $\coeffs(f,0)$ returns the constant term of $f$.
We then consider the remainder $u$ at Line \ref{line:a} which is the difference between $f-\varepsilon$ and its approximate SOHS decomposition.
As proved in Section \ref{sec:proofcsos1}, if the precision of root isolation is large enough, the stopping condition $\varepsilon > u_0 +
2\sum_{k=1}^{d}|u_k|$ is fulfilled, otherwise the precision
is increased.

To illustrate $\csosone$, we use the following simple
example.

\begin{example} Let $f = 5 + (1+i) z^{-1}+(1-i)z$ which is positive on
  $\ucircle$. We obtain $p=5+2x+2y$. With $\varepsilon = 1$, we check with
  $\hasrealrootoncircle$ that $p-\varepsilon$ is positive on $\ucircle$.
%
With precision $\delta = 16$, we compute complex approximation roots $\alpha_1 = -\frac{7}{4} - \frac{7}{4}i$ and $\bar \alpha_1$ of $f-\varepsilon$.
Defining $F=(z-\alpha_1)(z^{-1}-\bar\alpha_1)$, we obtain $a=\frac{32}{57}$, $u=f-\varepsilon-aF=(\frac{1}{57}+ \frac{i}{57}) z^{-1}+(\frac{1}{57} -\frac{i}{57})z$. Clearly, $\varepsilon = 1 > 0 + \frac{2\sqrt{2}}{57}$ so the condition in Line \ref{line:stop} is satisfied. Then, $f$ has an exact SOHS decomposition as follows:
\begin{center}
    $f = (1 -\frac{2\sqrt{2}}{57})+\frac{\sqrt{2}}{57}(z+
\frac{1+i}{\sqrt{2}})
(z^{-1} +\frac{1-i}{\sqrt{2}}) + \frac{32}{57}(z+\frac{7}{4}+ \frac{7}{4}i)(z^{-1}+\frac{7}{4} - \frac{7}{4}i)$.
\end{center}
\end{example}

\subsection{Proof of Theorem \ref{thm:bit1}}
\label{sec:proofcsos1}
\begin{proof}[Correctness of Algorithm $\csosone$]
We first prove that \textit{Algorithm $\csosone$ terminates and outputs an SOSH  decomposition of $f$.}

By Lemma \ref{lm:epsN}, there exits a positive rational $\varepsilon$ such that
$f - \varepsilon$ is also positive on $\ucircle$. Thus, the first loop
(from Line \ref{line:while0} to Line \ref{line:endwhile0}) of Algorithm
$\csosone$ terminates. The magnitude of the coefficients of the remainder
polynomial $u$ defined in Line \ref{line:lsu} converges to 0 as the precision
$\delta$ of the complex root finder goes to infinity (because of the continuity
of roots w.r.t. coefficients). This implies that the condition of Line
\ref{line:stop} is fulfilled after finitely many iterations, thus the second
loop (from Line \ref{line:while1} to Line \ref{line:endwhile1}) always
terminates. \if{ In Line \ref{line:lsu}, $aF$ is an approximation of
  $f-\varepsilon$ that depends on the accuracy $\delta$ of the computation in
  Line \ref{line:roots}. The difference $(f-\varepsilon) - aF$ goes to the zero
  polynomial when $\delta$ goes to the infinity. Thus, the sum $u_0 +
  2\sum_{k=1}^{d}|u_k|$ will go to zero as well.

This implies that the while loop in Lines \ref{line:while1}--\ref{line:endwhile1} will end after finitely many iterations.
\fi
Eventually, we have
    \begin{equation*}
f=\varepsilon+u_0+\left(u_1 {z}^{-1} +\bar{u}_1z\right)+\dots+\left(
 u_d {z}^{-d}+\bar{u}_dz^d\right)  +aF.
    \end{equation*}
In addition,
    \begin{equation}\label{eq:uk}
        u_k z^{-k} + \bar{u}_k  z^k =|u_k|\left( z^k
        +\frac{u_k}{|u_k|}\right) \left( z^{-k}
        +\frac{\bar{u}_k}{|u_k|}\right)-2|u_k|,
    \end{equation}
yielding \eqref{fz}.
From the proof of the Riesz-Fej\'{e}r theorem (see, e.g., \cite[pp. 3--5]{dumi2017}), $f$ can be decomposed as a single Hermitian square, thus the constant term of $f$ lies in $\Q_+$.
Similarly the constant term $\prod_{k=1}^d |\alpha_k|^2$ of $F$ is also a positive rational number, so $a \in \Q_+$.
    Clearly, the polynomial $F$ and the first
     term on the right-hand side of \eqref{eq:uk} are SOHS. Hence, as $\varepsilon > u_0 +
    2\sum_{k=1}^{d}|u_k|$, the right-hand side of \eqref{fz} is a  sum of
    $d+2$ Hermitian squares involving Gaussian (or Gaussian modulus) numbers.
\end{proof}

We now analyze the bit complexity of Algorithm $\csosone$.
\begin{proof}[Proof of Theorem~\ref{thm:bit1}] Let us show that \textit{the bitsizes of $u_0,\dots,u_d$, $\alpha_1,\dots,\alpha_d$ and $a$ in \eqref{fz} are bounded from above by $\widetilde{O}(d^5(d+\tau))$.}

The proof is almost the same as in the univariate real setting \cite[Theorem 4.3]{univsos}, thus we only provide the main ingredients and skip some technical details.
From Lemma \ref{lm:epsN}, there exists a positive
integer $N=\widetilde O (d^3(d+\tau))$ such that $f-\varepsilon$ is positive on $\ucircle$, with $\varepsilon =\frac{1}{2^{N}}$.
Define $m:=2d$ and
$g(z) := z^d(f-\varepsilon) = \bar f_d z^{2 d} + \dots + \bar f_1 z^{d+1} +
g_d z^d + f_1 z^{d-1} + \dots + f_d$, with $g_{d}=f_0-\varepsilon$.
Note that  $g$ and $f-\varepsilon$ have the same roots. Denote by $\zeta_1,\dots,
\zeta_{m}$ the (exact) complex roots of $g$ and by $\zeta'_1,\dots,\zeta'_{m}$
their approximations with a precision $\delta$, so that
$\zeta'_j=\zeta_j(1+e_j)$, where $|e_j|\leq e := 2^{-\delta}$, for
$j=1,\dots,m$.

Let $g':=\bar f_d(z-\zeta'_1)\dots(z-\zeta'_{m})$. The remainder polynomial $u$
defined in Line \ref{line:a} of Algorithm $\csosone$ satisfies $z^d u = g - g'$.

We now show that if $\delta=N+\log_2((2d+1)^2\|f\|)=\widetilde O (d^3(d+\tau))$,
then the coefficients of $u$ satisfy the condition $\varepsilon > u_0 +
2\sum_{k=1}^{d}|u_k|$. We apply Lemma \ref{lm:sepZ} to the polynomial obtained
by multiplying $z^d(f-\varepsilon)$ with the least common multiple of its
coefficients. Hence, we require an accuracy at least $\widetilde{O}(N d )$
to compute distinct approximations of its roots in the worst case. We will
actually need an accuracy of larger bitsize $\widetilde O (d^4(d+\tau))$ in the
worst case. \if{
\begin{equation}\label{eq:e}
    e=2^{-\delta}<
\frac{1}{\delta}<\frac{1}{Cd^3(d+\tau)}<\frac{1}{m(m+1)}.
\end{equation}
}\fi
Let $j \in \{0,1,\dots,d\}$.
Using Vieta's formulas, we have
\begin{equation}\label{eq:Vieta}
    \sum_{1\leq i_1<\cdots<i_{j}\leq m}\zeta_{i_1}\cdots
\zeta_{i_{j}} = (-1)^{j}\frac{g_{m-j}}{g_{m}} = (-1)^{j}\frac{g_{m-j}}{\bar f_d} \,.
\end{equation}
Similarly, we have:
\begin{equation}\label{eq:Vietaprime}
    \sum_{1\leq i_1<\cdots<i_{j}\leq m}\zeta'_{i_1}\cdots
\zeta'_{i_{j}} = (-1)^{j}\frac{g'_{m-j}}{\bar f_d} \,.
\end{equation}
We estimate an upper bound for the coefficient $\bar u_{d-j}$ of the difference
polynomial $u$. Clearly, $\bar u_{d-j}=g_{m-j}-g'_{m-j}$. From \eqref{eq:Vieta}
and \eqref{eq:Vietaprime}, we see that
\begin{align*}
|\bar u_{d-j}|  & =|\bar f_d| \left| \sum_{1\leq i_1<\cdots<i_{d+j}\leq m}(\zeta_{i_1}\cdots
\zeta_{i_{d+j}}-\zeta'_{i_1}\cdots \zeta'_{i_{d+j}}) \right|  \\
    &=|\bar f_d|  \left| \sum_{1\leq i_1<\cdots<i_{d+j}\leq m}\zeta_{i_1}\cdots
    \zeta_{i_{d+j}}\left( 1- \prod_{k=1}^{d+j} (1+e_{i_k})\right) \right |.
\end{align*}
Then, exactly as in the proof of \cite[Theorem 4.3]{univsos}, we rely on
\cite[Lemma 3.3]{higham2002accuracy} and the Cauchy bound \cite{Cauchy1830} to
obtain the desired estimates. \if{ As $e<\frac{1}{2m}$, we apply \cite[Lemma
  3.3]{higham2002accuracy} and get $(1+e_{i_1})\cdots(1+e_{i_{d+j}})\leq
  1+\theta_{d+j}$ with $|\theta_{i_{d+j}}|\leq \frac{i_{d+j}e}{1-i_{d+j}e}\leq
  \frac{me}{1-me}$. (\textit{Need to check again}) Since \eqref{eq:e}, we
  have $$(m+1)e-\frac{me}{1-me}=\frac{e(1-m(m+1)e)}{1-me}\geq 0.$$ This yields
  $\frac{me}{1-me}\leq (m+1)e$. So, we can conclude that
$$\left| 1-(1+e_{i_1})\cdots(1+e_{i_{d+j}})\right|\leq (m+1)e.$$
From the above presentation of $\bar u_j$ and the above results, we obtain the
following estimates:
\begin{equation}\label{eq:ginfty}
    |\bar u_{j}|\leq |\bar f_{d}|(j+1)e\leq \|f\|_{\infty}(m+1)e.
\end{equation}
The conclusion holds for
    $j=0,\dots, d$, one has
    $$u_0+2\sum_{k=1}^{d}|\bar u_k|\leq
    e(m+1)^2\|f\|_{\infty}\leq e(m+1)^2\|f\|_{\infty}.$$
It follows from $\delta=N+\log_2\left( (2d+1)^2\|g\|_{\infty}\right)$ that
    $e(2d+1)^2\|g\|_{\infty}=\varepsilon$. Therefore,
    $\varepsilon>u_0+2\sum_{k=1}^{d}|u_k|$ holds when
    $\delta=\widetilde{O}(d^3(d+\tau))$.

We choose $e_j=e=2^{-\delta}$ and $z'_j=z_j(1+2^{-\delta})$. This implies
    that
$|\bar u_{j}|= |\bar f_{d}||1-(1+2^{-\delta})^j|,$
    for all $j=0,\dots,d$.
It follows from $\hht(\bar f_d)\leq \tau$,
   $\hht(\delta)=\widetilde{O}(d^3(d+\tau))$,
     and
    $\hht(\varepsilon)=\widetilde{O}(d^3(d+\tau))$, that
    $$\hht(\bar u_{j})=\widetilde{O}(d^3(d+\tau)+jd^3(d+\tau))\leq \widetilde{O}(d^4(d+\tau)).$$
    Hence, the maximal bitsize of the coefficients of $u$ is bounded from above by $\widetilde{O}(d^4(d+\tau))$.

We now estimate the bit complexity of the coefficient $a$ in \eqref{fz}. From Line \ref{line:a} in $\csosone$, one has $a = \frac{f_0}{F_0}$, where $F_0$ is the constant term of $F$. Clearly, $\hht(f_0)\leq \tau$.  Because of \eqref{eq:Vieta}, $$F_0=(-1)^d\bar F_d\sum_{1\leq i_1<\cdots<i_{d}\leq m}z'_{i_1}\cdots z'_{i_{d}},$$
where $ z'_{i_j}\in\{\alpha_1,\dots,\alpha_{d},\frac{1}{\bar{\alpha}_1},\dots,\frac{1}{\bar{\alpha}_d}\}$. Since $\bar F_d =1$, $\hht(\alpha_i)\leq\hht(\delta)$, one has
 $\hht(F_0)\leq d\hht(\delta)+\log_2\binom{2d}{d}\leq d\hht(\delta)+d\log_2(d+1)=\widetilde{O}(d^4(d+\tau))$.

 Finally, the maximal bitsize of $u_k$'s and $a$ is bounded from above by $\widetilde{O}(d^4(d+\tau))$, as claimed.
}\fi

We now prove the remaining assertion in Theorem \ref{thm:bit1}: \textit{Algorithm $\csosone$ runs in $\widetilde{O}(d^6(d+\tau))$ boolean
    operations.}
Again the proof scheme is very similar as the one of  \cite[Theorem 4.4]{univsos}. The algorithm includes two steps.

We consider the first step checking that $g(z)$ defined in the previous proof has no real root on the unit circle.
Let $\varepsilon$ be given as in Lemma~\ref{lm:epsN} with  $\hht(\varepsilon)=\widetilde{O}(d^3(d+\tau))$.
By relying on Sylvester--Habicht sequences \cite[Corollary~5.2]{lick2001}, the check can be performed
using $O((2d)^2\hht(\varepsilon))=\widetilde{O}(d^5(d+\tau))$ boolean operations.
%
In the second step, we compute approximate complex roots of $g(z)$ and check the condition at Line \ref{line:stop}.
It follows from \cite[Theorem~4]{mehlhorn2015approximate} that isolating disks of radius less than $2^{-\delta}$ for all complex roots of $g(z)$ can be computed in $\widetilde{O}(d^3 + d^2 \hht(\varepsilon) + d \delta)) = \widetilde{O}(d^6(d+\tau))$ boolean operations.
The computation of all $u_k$ has a negligible cost w.r.t. to the computation of the complex roots.
Therefore, we conclude that  $\csosone$ runs in $\widetilde{O}(d^6(d+\tau))$
boolean operations.
\end{proof}

\section{Algorithm based on complex SDP}\label{sec:csos2}

This section states and analyzes another perturbation-compensation algorithm,
named $\csostwo$, to compute an SOHS decomposition of a trigonometric polynomial
being positive on $\ucircle$. In the algorithm, the approximate SOHS
decomposition for the perturbation is computed by using complex SDP solving. It
can be viewed as the adaptation of the procedure $\texttt{intsos}$ (stated and
analyzed in \cite[\S~3]{magron2021exact}) to the complex univariate setting. Let
$I$ stands for the identity matrix of size $d+1$. A Hermitian matrix $Q$ is said
to be positive semidefinite (resp. definite) if $Q$ has only non-negative (resp.
positive) eigenvalues, and in this case we use the notation $Q \succeq 0$ (resp.
$Q \succ 0$). Given $f \in \Hi[z]$ of degree $d$, recall that a Hermitian matrix
$Q \in \C^{(d+1) \times (d+1)}$ is called a \emph{Gram} matrix associated with
$f$ if $f= v_d^{\star} \cdot Q \cdot v_d$, where $v_d(z) := (1,z,\dots,z^d)$
contains the canonical basis for polynomials of degree $d$ in $z$. By
\cite[Theorem 2.5]{dumi2017}, $f$ is positive on $\ucircle$ if and only if there
exists a positive definite Gram matrix associated to $f$.

\subsection{Description}
The input of Algorithm $\csostwo$ includes a polynomial $f\in
 \Hi(\Z)[z]$ of
degree $d$ which is positive on $\ucircle$.
The outputs are $\varepsilon \in \Q_+$, a list of Gaussian numbers
$[u_0,u_1,\dots,u_d]$, and a list of polynomials $[s_1,\dots,s_d]$ in $\Q[i][z]$ providing an
SOHS decomposition of $f$ as follows
    {\small
        \begin{equation}\label{eq:fsos2}
            f = \Big(\varepsilon-u_0-2\sum_{k=1}^{d}|u_k|\Big)
            +\sum_{k=1}^{d}|u_k|\Big(z^k
            +\frac{u_k}{|u_k|}\Big)
            \Big(z^{-k} +\frac{\bar{u}_k}{|u_k|}\Big)
            +\sum_{k=0}^{d}
            s_k^\star s_k.
        \end{equation}
    }

\begin{algorithm}
\caption{$\csostwo$}\label{alg:csos2}
\begin{algorithmic}[1]
\Require  $f \in \Hi(\Z)[z]$ positive on $\ucircle$ of degree $d$
\Ensure $\varepsilon\in \Q_+$,  $[u_0,u_1,\dots,u_d]$
    in $\Q[i]$, $[s_0,\dots,s_d]$ in $\Q[i][z]$ providing an
    SOHS decomposition of $f$ as in \eqref{eq:fsos2}.

\State  $\delta \gets 1$, $R \gets 1$, $\delta_c = 1$, $\varepsilon := 1$ and $p\gets f(x+iy)$\label{line:sub2}
\While {$\hasrealrootoncircle(p-\varepsilon)$} $\varepsilon \gets \frac{\varepsilon}{2}$
\EndWhile
\State boo := false
\While {not boo} \label{line:while4}
\State $(\tilde{Q}, \tilde{\lambda}) \gets
    \sdpfun{f-\varepsilon}{ \delta}{R}$ \label{line:sdp}
\State $[s_0,\dots,s_d]
    \gets \choleskyfun{\tilde{Q}}{\tilde{\lambda}}{\delta_c}$
    \label{line:chol} \Comment{$f_\varepsilon \simeq
        \sum_{k=0}^{d}  s_k^\star s_k$}
\State $u: = f-\varepsilon - \sum_{k=0}^{d}
    s_k^\star s_k$, \ $[u_0,u_1,\dots,u_d] \gets \coeffs(u)$\label{line:lsu2}
\If {$\varepsilon> u_0 + 2\sum_{k=1}^{d}|u_k|$}\label{line:condition}
    boo := true
\Else \ $\delta:=2\delta, R:=2R, \delta_c:=2\delta_c$
\EndIf
\EndWhile\label{line:endwhile4}
\State \Return $\varepsilon, [u_0,u_1,\dots,u_d], [s_0,\dots,s_d]$
\end{algorithmic}
\end{algorithm}
The first steps of $\csostwo$ (Lines
\ref{line:sub2}--\ref{line:endwhile0}) are exactly the same  as $\csosone$ to obtain $\varepsilon \in \Q_+$ such that $f - \varepsilon$ is positive on $\ucircle$.
Then, instead of using root
isolation as in $\csosone$, $\csostwo$ relies on complex SDP  (Line \ref{line:sdp}) and Cholesky's decomposition (Line \ref{line:chol}) to compute an approximate SOHS decomposition of the perturbed polynomial.
With $f-\varepsilon$, $\delta$, and $R$, the
$\sdp$ function calls an SDP solver to compute a rational approximation
$\tilde{Q}$ of the Gram
matrix associated to $f-\varepsilon$ and a rational approximation
$\tilde{\lambda}$ of its smallest eigenvalue.
As in \cite{magron2021exact} we analyze the complexity of the procedure by assuming that \texttt{sdp}
relies on the ellipsoid algorithm  \cite{GroetschelLovaszSchrijver93}, running in polynomial-time within a given accuracy $\delta$ and a radius bound $R$ on the Frobenius norm of $\tilde{Q}$.
Its outputs are obtained by solving the following complex SDP:
\begin{align}\label{eq:posGram}
    \lambda_{\min} = \displaystyle\max_{Q,\lambda} & \  \lambda  \nonumber\\
    \mbox{s.t.}& \ \trace (\Theta_k Q) = f_k - 1_{k=0}  \varepsilon \,, k = -d,\dots,d \,, \\
    & \ Q \succeq \lambda I \,, \lambda \geq 0 \,, \ Q \in \C^{(d+1) \times (d+1)} \nonumber\,,
\end{align}
where $\Theta_k$ is the elementary Toeplitz matrix with ones on the $k-$th
diagonal
and zeros elsewhere, $1_{k=0} = 1$ if $k=0$ and $0$ otherwise, $\trace(\cdot)$ stands for the usual matrix trace operator.
The equality constraints of SDP \eqref{eq:posGram} corresponds to the relation $f(z) - \varepsilon= v_d^{\star}(z) \cdot Q \cdot v_d (z) $.
This SDP program (corresponding to SDP (2.14) in \cite{dumi2017}) computes the Gram matrix associated to $f$ with the largest minimal eigenvalue.
The $\cholesky$ function computes first an approximate Cholesky's decomposition $L L^T$ of $\tilde{Q}$ with precision $\delta_c$, and provides as output a list of polynomials $[s_0,\dots,s_d] \in \Q[i][z]$, $s_k$ is the inner product of the $(k+1)$-th row of $L$ by $v_d$.
One would expect to have $f-\varepsilon = \sum_{k=0}^d s_k^\star s_k$ after using exact SDP and Cholesky's decomposition.
Since the SDP solver is not exact, we have to consider the remainder $u = f - \varepsilon - \sum_{k=0}^d s_k^\star s_k$ and proceed exactly as in $\csosone$ to obtain an exact SOHS decomposition.

\subsection{Proof of Theorem \ref{thm:bit2}}

The lemma below prepares the bit complexity analysis of $\csostwo$.
\begin{lemma}\label{lm:deltaR}
    Let $f\in \Hi(\Z)[z]$ be positive on $\ucircle$
    of degree $d$ and bitsize $\tau$. Assume that $Q$ is a positive definite
     Gram matrix associated to $f$.
Then, there exist $\varepsilon\in \Q_+$ of bitsize $\widetilde O(d^3(d+\tau))$ such  that
     $f-\varepsilon$ is positive on $\ucircle$,
$\delta$ of bitsize $\widetilde O(d^3(d+\tau))$ and $R$ of bitsize $O (\hht(d)+\tau)$ such that $Q-\varepsilon/(d+1) I$ is a Gram
matrix associated to $f-\varepsilon$ with $Q-\varepsilon/(d+1) I \succ
 2^{-\delta}I$ and $\sqrt{\trace((Q-\varepsilon I)^2)}\leq R$.
\end{lemma}
\begin{proof}
By Lemma \ref{lm:epsN}, there is a positive integer
 $N$ and $\varepsilon = 2^{-N} =\widetilde{O}(d^3(d+\tau))$ such that $ f- 3\varepsilon/ 2 > 0$ on $\ucircle$.
Let $\delta:= \lceil N+1 + \log_2(d+1) \rceil =\widetilde{O}(d^3(d+\tau))$ so that $2^{-\delta} \leq \frac{\varepsilon}{2(d+1)}$.
One has $v_d^\star(z) v_d(z) = d+1$, thus $f(z)- \varepsilon = v_d^\star(z) (Q-\varepsilon/(d+1) I) v_d(z)$.
Since $f(z)- \varepsilon - 2^{-\delta} (d+1) > 0$, we obtain $v_d^\star(z) (Q-\varepsilon/(d+1) I - 2^{-\delta}I) v_d(z) > 0$.

Let $R:=\sqrt{(d+1)}f_0$.
Note that the equality constraint of SDP \eqref{eq:posGram} with $k=0$ reads $\trace (Q) = f_0 - \varepsilon \leq f_0$.
The maximal eigenvalue of $Q$ is less than $f_0$ and $\trace((Q-\varepsilon I)^2) \leq \trace(Q^2) \leq (d+1) f_0^2= R^2$.
\end{proof}
%
\if{
From this pair, we compute an approximate Cholesky's decomposition $L L^T$ of
$\tilde{G}$ with precision $\delta_c$. Again, when the precision
parameters $\delta, \delta_c, R$ go to the infinity together,
$u=(f-\varepsilon) - \sum_{j=0}^{d}s_j^\star s_j$ tends to the zero polynomial.
}\fi


The correctness and bit complexity proofs are very similar to the ones for the
$\texttt{intsos}$ algorithm in \cite[Proposition 10]{magron2021exact}, so we
only provide a sketch with the main ingredients.
\begin{proof}[Proof of correctness for $\csostwo$]
Since $f- \varepsilon$ is positive on $\ucircle$, SDP
\eqref{eq:posGram} has always a strictly feasible solution  for precision parameters ($\delta, R$) with bitsizes as in Lemma \ref{lm:deltaR} and
the $\sdp$ function returns an approximate Gram matrix $\tilde{Q}$ of
  $f-\varepsilon$ such that $\tilde{Q}\succeq 2^{-\delta}I$ and $\trace(Q^2)
\leq R^2$.
 In parti\-cular, we obtain a rational approximation
$\tilde{\lambda}\geq 2^{-\delta}$ of the smallest eigenvalue of $\tilde{Q}$.

At Line \ref{line:sdp}, we compute an approximate Cholesky decomposition of $\tilde{Q}$
 by using the $\cholesky$ procedure; we obtain a rational nonsingular factor if there exists $\delta_c$ satisfying \eqref{eq:2deltac}.
Let $\delta_c$ be the smallest integer such that $2^{-\delta_c} <
 \frac{2^{-\delta}}{d^2+d+(d-1)2^{-\delta}}$.
 Since the
  function $t \mapsto \frac{t}{d^2+d+(d-1)t}$ increases on $[0,+\infty)$ and
   $\tilde{\lambda} \geq 2^{-\delta}$, \eqref{eq:2deltac} holds.

We now consider the polynomial $u=f-\varepsilon -
\sum_{k=0}^{d}s_k^\star s_k$. The while loop (Lines \ref{line:while4}--\ref{line:endwhile4}) terminates when $\varepsilon>
 u_0 + 2\sum_{k=1}^{d}|u_k|$. This condition holds if $|u_k|\leq
  \frac{\varepsilon}{2d+1}$, for all $k=0,\dots,d$.
As in the proof of \cite[Proposition 10]{magron2021exact}, we prove that this holds for large enough $\delta$ and $\delta_c$, with bitsizes $\widetilde{O}(d^3(d+\tau))$.
\if{
We point out that the last
condition holds when $\delta$ and $\delta_c$ are both large
 enough. Indeed, we have
\begin{equation*}\label{eq:ukrk}
     u_k = f_k - \varepsilon_k- \left( \sum_{j=0}^{d}s_j^\star s_j\right)_k, \
      k=-d,\dots,d,
 \end{equation*}
where
$\varepsilon_0 = \varepsilon$, $\varepsilon_k = 0$ for $k\neq 0$, and $(\sum_{j=0}^{d}s_j^\star s_j)_k$ is the coefficient of monomial $z^k$ in the involved polynomial.
\if
$u_k= \begin{cases}
f_0 - \varepsilon - (\sum_{j=0}^{d}s_j^\star s_j)_0 &\mbox{if } k = 0, \\
f_k - (\sum_{j=0}^{d}s_j^\star s_j)_k & \mbox{if } k \neq 0.
 \end{cases}$
\fi
The positive definite matrix $\tilde{G}$ computed by the SDP solver is an
 approximation of the Gram matrix associated to $f-\varepsilon$. With the
  precision $\delta$, from \eqref{eq:uk}, we see that $\tilde{G}\succeq
   2^{-\delta}I$ and
\begin{equation*}\label{eq:rkTheta}
    \rvert f_k - \varepsilon_k- \trace(\Theta_k\tilde{G})\lvert=\rvert f_k - \varepsilon_k- \sum_{i+j=k}\tilde{G}_{ij}\lvert\leq
 2^{-\delta}.
\end{equation*}

Furthermore, from \eqref{eq:H}, the approximate Cholesky decomposition
 $LL^T$ of $\tilde{G}$
 performed at precision $\delta$ satisfies $LL^T = \tilde{G}+H$ and
$$|H_{ij}|\leq\frac{(d+2)2^{-\delta_c}
    \sqrt{|\tilde{G}_{ii}\tilde{G}_{jj}|}}{1-(d+2)2^{\delta_c}},$$
for all $i,j$ in $\{-d,\dots,d\}$. Applying the Cauchy-Schwarz inequality
for the trace function, we see that
$$\sum_{k=-d}^d|\tilde{G}_{kk}| = \trace(\tilde{G})\leq
 \sqrt{\trace(\tilde{G}^2)}\sqrt{\trace(I)}\leq R\sqrt{d+1}.$$
So, for each $k$ in $\{-d,\dots,d\}$, we have
\begin{equation}\label{eq:Rd1}
\Big|\sum_{i+j=k}\sqrt{\tilde{G}_{ii}\tilde{G}_{jj}}\Big|\leq \sum_{i+j=k} \frac{\tilde{G}_{ii}+\tilde{G}_{jj}}{2}\leq \trace(\tilde{G})\leq R\sqrt{d+1}.
\end{equation}
Therefore, we have
$$\Big|\sum_{i+j=k}\tilde{G}_{ij} - \Big(\sum_{j=0}^{d}s_j^\star s_j\Big)_k\Big| =\Big| \sum_{i+j=k}\tilde{G}_{ij} - \sum_{i+j=k}\big(LL^T\big)_k\Big| = \Big| \sum_{i+j=k}H_{ij}\Big|;$$
it follows from \eqref{eq:H} and \eqref{eq:Rd1} that the last number is bounded by
$$\frac{(d+2)2^{-\delta_c}}{1-(d+2)2^{-\delta_c}}\sum_{i+j=k}\sqrt{|\tilde{G}_{ii}\tilde{G}_{jj}|}\leq \frac{R\sqrt{d+1}(d+2)2^{-\delta_c}}{1-(d+2)2^{-\delta_c}}.$$
Let us take the smallest $\delta$ such that $2^{-\delta}\leq \frac{\varepsilon}{2(2d+1)} = \frac{1}{(2d+1)2^{N+1}}$ as well as the smallest $\delta_c$ such that
$$\frac{R\sqrt{d+1}(d+2)2^{-\delta_c}}{1-(d+2)2^{-\delta_c}}\leq \frac{\varepsilon}{2(2d+1)},$$
i.e., $\delta=\lceil N+1+\log_2(2d+1)\rceil$ and $\delta_c=\lceil\log_2R+\log_2(d+2)+\log_2(2^{N+1}(2d+1)^{3/2}+1)\rceil$.
It implies from above inequalities that
\begin{align*}
|u_k| & \leq \Big|f_k - \varepsilon_k- \sum_{i+j=k}\tilde{G}_{ij}\Big| + \Big|\sum_{i+j=k}\tilde{G}_{ij}-\Big(\sum_{j=0}^{d}s_j^\star s_j\Big)_k\Big|\\
& \leq  \frac{\varepsilon}{2(2d+1)}  + \frac{\varepsilon}{2(2d+1)}=\frac{\varepsilon}{2d+1}.
\end{align*}
This guarantees that the second while loop terminates.

From Lemma~\ref{lm:deltaR} with $N= \widetilde{O}(d^3(d+\tau))$ and $R =\widetilde{O}(d^4(d+\tau))$, and $\delta,\delta_c$ defined as above, one has $\delta = \widetilde{O}(d^3(d+\tau))=\delta_c$. This number is also a bound of the bitsizes of the coefficients involved in the approximation $\sum_{j=0}^{d}s_j^\star s_j$ as well as the coefficients
of $u$. Thus, look at the constant term in \eqref{eq:fsos2}, namely $\varepsilon-u_0-2\sum_{k=1}^{d}|u_k|$, its maximal bitsize is bounded by $\widetilde{O}(d^3(d+\tau))$.
The proof is complete.
}\fi
\end{proof}
%
\vspace*{-0.3cm}
\begin{proof}[Bit complexity estimate for $\csostwo$] {We prove now that
    Algorithm $\csostwo$ runs in $\widetilde{O}(d^{13}(d+\tau)^2)$ boolean
    operations.}

Assume that $\varepsilon,\delta, R$ and $\delta_c$ are given as above so that, before terminating, Algorithm $\csostwo$ performs a single iteration in each while loop.
From above results, the bitsize of each $\varepsilon,\delta,\delta_c$ is upper bounded by $\widetilde{O}(d^3(d+\tau))$ and that of $R$ is $O(\hht(d)+\tau)$.

\if
Mimic/adapt the proof of Proposition 10 and Theorem 11 from
 \cite{magron2021exact} to do the bit complexity analysis. As for the intsos
algorithm from \cite{magron2021exact},
\fi

To investigate the computational cost of the call to $\sdp$ at
Line~\lineref{line:sdp},  let us note $\nsdp = d+1$ the size of $\tilde{Q}$ and $\msdp = 2d+1$ the number of affine constraints of SDP  \eqref{eq:posGram}.
We rely on the bit complexity analysis
  of the ellipsoid method 
   \cite{porkolab1997complexity}.
Solving SDP~\eqref{eq:posGram} is performed in
${O}( \nsdp^4 \log_2 (2^\tau \nsdp \, R \, 2^\delta) )$ iterations of the
 ellipsoid method, where each iteration requires
${O}(\nsdp^2(\msdp+\nsdp) )$ arithmetic operations over
$\log_2 (2^\tau \nsdp \, R \, 2^\delta)$-bit numbers
(see, e.g.,  \cite{GroetschelLovaszSchrijver93,porkolab1997complexity}).
We obtain the following estimates:
$O(\nsdp^4 \log_2(2^\tau \nsdp \, R \, 2^\delta)) = \widetilde{O}(d^7(d+\tau))$,
 $O(\nsdp^2(\msdp+\nsdp)) = O(d^{3})$, and $O(\log_2 (2^\tau \nsdp \, R \,
  2^\delta)) = \widetilde{O}(d^3(d+\tau))$.
Therefore, the ellipsoid algorithm runs in boolean time
$\widetilde{O}(d^{13}(d+\tau)^2)$ to compute the approximate Gram matrix
    $\tilde{Q}$.

Next, we compute the cost of calling $\cholesky$ in Line \ref{line:chol}.
Note that Cholesky's decomposition is performed in $O(\nsdp^3)$ arithmetic
 operations over $\delta_c$-bit numbers. Because of $\delta_c=\widetilde{O}(d^3(d+\tau))$ and $\nsdp=d+1$, $\cholesky$  runs in boolean time $\widetilde{O}(d^{6}(d+\tau))$.

 The other elementary arithmetic operations of Algorithm $\csostwo$ have a negligible cost w.r.t. to the $\sdp$ procedure. Hence, the algorithm runs in boolean time $\widetilde{O}(d^{13}(d+\tau)^2)$.
\end{proof}

\section{Rounding-projection algorithm}\label{sec:csos3}
Here, we introduce Algorithm $\csosthree$ which is an adaptation of
the rounding-projection algorithm by Peyrl and Parrilo, stated in
\cite{peyrl2008} and analyzed in \cite[\S~3.4]{magron2021exact}, and investigate
its bit complexity.

The input of $\csosthree$ is a polynomial $f \in \Hi(\Z)[z]$ of degree
$d$ which is positive over $\ucircle$.
The outputs consist of a list $[c_0,\dots,c_d] \subset \Q_+$ and a list of
polynomials $[s_0,\dots,s_d]$ in $\Q[i][z]$ that provide an SOHS decomposition of
$f$, namely $f = \sum_{k=0}^d c_k s_k^\star s_k$.

\begin{algorithm}
    \caption{{$\csosthree$}}
    \label{alg:PP}
    \begin{algorithmic}[1]
\Require $f \in \Hi(\Z)[z]$ positive on $\ucircle$ of degree $d$
        \Ensure lists $[c_0,\dots, c_d] \subset \Q_+$ and  $[s_0,\dots,s_d ] \subset \Q[i][z]$ providing an
        SOHS decomposition of $f$ as follows: $$f = \sum_{k=0}^d c_k s_k^\star s_k$$

        \State  $\delta \gets 1$, $R \gets 1$, $\delta_c = 1$, $\hat{\delta} := 1$
        \State boo := false
        \While {not boo}\label{line:3Wh}
        \State $(\tilde Q, \tilde \lambda) \gets \sdpfun{f}{\delta}{R}$
        \If {$\tilde \lambda > 0$} boo := true
        \Else $\ \delta := 2 \delta$, $R := 2 R$
        \EndIf
        \EndWhile\label{line:3endWh}
        \State boo := false
        \While {not boo}\label{line:while5}
        \State $\hat Q \gets \roundfun{\tilde Q}{\hat \delta}$
         \label{line:roundPP}
        \For {$j \in \{0, \dots, d \}, k \in \{0,\dots,j \} $}\label{line:forQ}
        \State $Q_{j,j-k} := \hat Q_{j,j-k} - \frac{1}{d-k+1} (\sum_{i=k}^d
         \hat Q_{i,i-k} -  f_k )$ \label{line:projPP}
        \State $Q_{j-k,j} := Q_{j,j-k}^\star$
        \EndFor\label{line:endforQ}
        \State $[c_0,\dots, c_d; s_0,\dots,s_d] \gets \ldlfun{Q}$
         \label{line:cholPP} \Comment{$f =  \sum_{k=0}^d c_k s_k^\star  s_k$}
        \If {$c_0,\dots,c_d \in \Q_+, s_0,\dots,s_d \in \Q[i][z]$}
        boo := true \label{line:LDLok}
        \Else $\ \hat \delta \gets 2 \hat \delta$
        \EndIf
        \EndWhile \label{line:PPf}
        \State \Return $[c_0,\dots,c_d]$, $[s_0,\dots,s_d]$
    \end{algorithmic}
\end{algorithm}

As in $\csostwo$, the first while loop from Lines \ref{line:3Wh}--\ref{line:3endWh} provides an approximate Gram matrix $\tilde Q$ associated to $f$ and an approximation $\tilde \lambda$ of its smallest eigenvalue.
In Line~\ref{line:roundPP}, we round the matrix $\tilde Q$ up to precision $\hat\delta$ to obtain a  matrix $\hat Q$, with Gaussian coefficient entries.
The for loop from Line \ref{line:forQ} to Line \ref{line:endforQ} is the projection step to ensure that the equality constraints of SDP \eqref{eq:posGram} hold exactly.
Then we compute the $LDL^T$ decomposition of $Q$.
The list $[c_0,\dots,c_d]$ is the list of coefficients of the diagonal matrix $D$ and each $s_k$ is the inner product of the $(k+1)$-th row of $L$ and the vector $v_d$ of all monomials up to degree $d$.
If all $c_k$'s are positive rationals and all polynomials $s_k$' have Gaussian coefficients, then the second while loop ends, otherwise we increase the precision $\hat \delta$.

\if
Mimic/adapt Section 3.4 in \cite{magron2021exact} to get the bit complexity estimate.
\fi

As emphasized in \cite[\S~3.4]{magron2021exact}, it turns out that $\csostwo$ and $\csosthree$ have the same bit complexity.
We omit any technicalities as the proof is almost the same as \cite[Theorem~12]{magron2021exact}.
\if{
\begin{proposition}\label{prop:deltaR} For $f \in \Hi(\Z)[z]$ positive on $\ucircle$, with degree $d$ and bitsize $\tau$, there exist $\delta$, $\hat\delta$ upper bounded by $\widetilde{O}(d^3(d+\tau))$, and $ R$ upper bounded by $O(\hht(d)+\tau)$ such that Algorithm
$\csosthree$ outputs an SOHS decomposition of $f$.
\end{proposition}

\begin{proof} Suppose that the matrix $Q \succ 0$ obtained after the second while loop (Lines \ref{line:while5}--\ref{line:PPf}) which is associated to $f$ with the smallest eigenvalue $\lambda$. Let $N\in \N$ be the smallest integer satisfying $2^{-N}\leq \lambda$. From Lemma~\ref{lm:deltaR}, the maximal bitsize of $N$ is upper bounded by $\widetilde{O}(d^3(d+\tau))$. By Proposition 8 in \cite{peyrl2008}, Algorithm $\csosthree$ terminates and outputs such a matrix $Q$ together with an SOHS decomposition of $f$ if $2^{-\hat\delta}+2^{-\delta_E}\leq 2^{-N}$, where $\delta_E$ stands for the Euclidean distance between $Q$ and $\hat Q$, yielding   $$\sqrt{\sum_{i,j}(Q_{ij}-\hat Q_{ij})^2} = 2^{-\delta_E}.$$

For all index pair $i,j\in \{-d,\dots,d\}$, one has $|\hat Q_{ij}- \tilde Q_{ij}|\leq 2^{-\hat\delta}$.
As in the  proof of part $\rm a)$ in Theorem~\ref{thm:bit2}, at the SDP precision $\delta$, one has $\tilde Q \succeq 2^{-\delta}I$ and
$$\big|f_k- \sum_{i+j=k}\tilde Q_{ij}\big| \leq 2^{-\delta}.$$
We denote $e_{ij}:=\big( \sum_{i'+j'=i+j}\hat Q_{i'j'}\big) -f_{i+j}= \sum_{i'+j'=i+j}\big( \hat Q_{i'j'}-Q_{i'j'}\big)$, for $i,j\in \{-d,\dots,d\}$. The cardinal of the set $\{(i',j'):i'+j'=i+j\}$ is $d+1-(i+j)$. Hence, we have the following estimate:
\begin{align*}
    |e_{ij}| & \leq \sum_{i'+j'=i+j}\big|\hat Q_{i'j'}-\tilde Q_{i'j'}\big|+ \big|\sum_{i'+j'=i+j}\tilde Q_{i'j'}-f_{i+j}\big| \\
    & \leq  (d+1-i-j)2^{-\hat\delta} + 2^{-\delta} \leq  (d+1-i-j)(2^{-\hat\delta} + 2^{-\delta}).
\end{align*}
It follows that
$$2^{-\delta_E} = \sum_{i,j}\frac{e_{ij}}{d+1-i-j}\leq 2(d+1)^d(2^{-\hat\delta} + 2^{-\delta}).$$
To ensure that $2^{-\hat\delta}+2^{-\delta_E}\leq 2^{-N}$, it is sufficient to have
$$2^{-\hat\delta} +2(d+1)^d(2^{-\hat\delta} + 2^{-\delta})\leq 2^{-N}.$$
Since $\hht(N)=\widetilde{O}(d^3(d+\tau))$, the last inequality is obtained if $\delta$ and $\hat\delta$ are upper bounded by $\widetilde{O}(d^2(d+\tau))$.

Clearly, the $\sdp$ function is successful if $R$ is upper bounded by $O(d^4(d+\tau))$ as Lemma \ref{lm:deltaR}.
\end{proof}
}\fi

\vspace*{-0.3cm}

\begin{proof}[Bit complexity estimate for $\csosthree$]
  For $f \in \Hi(\Z)[z]$ positive over $\ucircle$ with degree $d$ and maximal
  bitsize $\tau$, there exist $\delta$, $\hat\delta$ with bitsizes upper bounded
  by $\widetilde{O}(d^3(d+\tau))$, and $ R$ with bitsize upper bounded by
  $O(\hht(d)+\tau)$ such that Algorithm $\csosthree$ outputs an SOHS
  decomposition of $f$.
  The bitsize of the output coefficients is
  upper bounded by the output bitsize of the $LDL^T$ decomposition of the matrix
  $Q$, that is $O(\hat\delta (d+1)^3) =\widetilde{O}(d^6(d+\tau))$. The running
  time is estimated as for $\csostwo$.
\end{proof}

\section{Practical experiments}
This section is dedicated to experimental results obtained by running our three certification algorithms, $\csosone$, $\csostwo$ and $\csosthree$, stated in Section \ref{sec:csos1}, \ref{sec:csos2} and \ref{sec:csos3} respectively.
First, we compare their performance to certify that trigonometric polynomials with Gaussian coefficients are positive on the unit circle $\ucircle$.
Next, we describe how to extend our third algorithm $\csosthree$ to design a finite impulse response (FIR) filter in a certified fashion.
Our code is implemented in {\sc Julia}, freely available online \footnote{\url{https://homepages.laas.fr/vmagron/files/csos.zip}}, and the results are obtained on an Intel Xeon 6244 CPU (3.6GHz) with 1.5 TB of RAM.
In $\csosone$ and $\csostwo$, we compute $\varepsilon$ such
that $f-\varepsilon$ is positive on $\ucircle$ in Lines \ref{line:while0}--\ref{line:endwhile0} by using {\sc msolve} \cite{msolve}  within the Julia library $\mathtt{GroebnerBasis.jl}$.
The corresponding running time is denoted by
$t_{\varepsilon}$.
We denote by $t_{u}$ the running time spent to compute the remainder polynomial $u$ and to perform the comparison involving its coefficients and $\varepsilon$.
In $\csosone$, we compute approximate roots of $f-\varepsilon$ with the arbitrary-precision library $\mathtt{PolynomialRoots.jl}$ \cite{skowron2012general}.
In $\csostwo$ and $\csosthree$, we model SDP \eqref{eq:posGram} via {\sc JuMP} \cite{dunning2017jump} and solve it with Mosek \cite{andersen2000mosek}.
Exact arithmetic is performed with the {\sc Calcium} library available in $\mathtt{Nemo.jl}$.

\subsection{Positivity verification}
\if The first family contains polynomials having real coefficients $f_d(z)=1+2d+\sum_{k=1}^{d}( z^k+\frac{1}{z^k}).$
Because of
$z^k+z^{-k}=2 \cos(k\varphi) \geq -2$ for $z \in \ucircle$, $f_d$ is
 positive on $\ucircle$.
 \fi
We consider a family of trigonometric polynomials having Gaussian integer coefficients $f_d = 10d + \sum_{k=1}^{d}((1-i)z^{-k}+(1+i)z^k)$.
For each $d \in \{50,100,150,200,250\}$, we prove that $f_d$ is positive on $\ucircle$ by computing via $\csosone$, $\csostwo$ and $\csosthree$ exact SOHS decompositions.
Note that each such $f_d$ is positive on $\ucircle$ since $ z^{-k}+z^k \geq -2$.

For $\csosone$, we use a precision $\delta = 64$ (bits) to isolate complex roots.
As a side note, we were not able to use arbitrary-precision SDP solvers (e.g., SDPA-GMP) within $\csostwo$ and $\csosthree$, because {\sc JuMP} only allows us to rely on double floating-point arithmetic at the moment.
The running times (in seconds) of the 3 algorithms are reported in Table \ref{tbl:table}.
As expected from the theoretical bit complexity results from Theorem \ref{thm:bit1} and Theorem \ref{thm:bit2}, Algorithm $\csosone$ performs better than $\csostwo$ and $\csosthree$.
The reason why $\csostwo$ is faster than $\csosthree$ is due to the fact that the latter algorithm requires to perform  an exact Cholesky's factorization.
Even though $\csosone$ happens to be the best choice to verify the positivity of polynomials with known coefficients, the use of an SDP solver is mandatory to optimize over positive polynomials with unknown coefficients, as demonstrated in the next subsection.

\vspace*{-0.2cm}

\begin{table}[ht!]
\begin{tabular}{c|ccc|ccc|c|}
    \cline{2-8}
    & \multicolumn{3}{c|}{$\csosone$}                                               & \multicolumn{3}{c|}{$\csostwo$}                                               & $\csosthree$ \\ \hline
    \multicolumn{1}{|c|}{$d$} & \multicolumn{1}{c|}{$t_{\varepsilon}$} & \multicolumn{1}{c|}{$t_{u}$} & total & \multicolumn{1}{c|}{$t_{\varepsilon}$} & \multicolumn{1}{c|}{$t_{u}$} & total & total        \\ \hline
    \multicolumn{1}{|c|}{50}  & \multicolumn{1}{c|}{0.2}               & \multicolumn{1}{c|}{0.3}     & 0.6   & \multicolumn{1}{c|}{0.2}               & \multicolumn{1}{c|}{6.6}      & 6.8    & 7.7           \\ \hline
    \multicolumn{1}{|c|}{100} & \multicolumn{1}{c|}{1.6}               & \multicolumn{1}{c|}{2.9}     & 4.5   & \multicolumn{1}{c|}{1.6}               & \multicolumn{1}{c|}{128}     & 130   & 184          \\ \hline
    \multicolumn{1}{|c|}{150} & \multicolumn{1}{c|}{5.2}               & \multicolumn{1}{c|}{13}     & 19   & \multicolumn{1}{c|}{5.2}               & \multicolumn{1}{c|}{830}     & 838   & 1460         \\ \hline
    \multicolumn{1}{|c|}{200} & \multicolumn{1}{c|}{24}                & \multicolumn{1}{c|}{26}      & 51    & \multicolumn{1}{c|}{24}                & \multicolumn{1}{c|}{3460}    & 3485  & 7214         \\ \hline
\multicolumn{1}{|c|}{250} & \multicolumn{1}{c|}{64}                &  \multicolumn{1}{c|}{55}      & 120   & \multicolumn{1}{c|}{64}                & \multicolumn{1}{c|}{10553}    & 10622   &   24852      \\ \hline
\end{tabular}
\caption{Performance of Algorithms $\csosone$, $\csostwo$, and $\csosthree$}
\label{tbl:table}
\end{table}


\vspace*{-0.6cm}

\subsection{Design of a certified linear-phase FIR filter}\label{subsec:filter}

This section is devoted to the design of a linear-phase finite impulse response (FIR) filter.
This boils down to solving an energy minimization problem.
To obtain a certified filter, we first solve a semidefinite optimization  problem (corresponding to SDP (5.12) from \cite{dumi2017}) and transform the numerical output into an exact certificate via a projection procedure similar to the one used in $\csosthree$.

Let $H(z)=\sum_{k=-d}^{d} h_kz^{-k}$ be an FIR filter of order $d$, with real coefficients.
Let $h = [h_0,\dots,h_d]$ be the coefficient vector of $H$.
Since we work on the unit circle, we have $z = \exp(i \omega)$, for $\omega \in \R$, and we abuse notation by writing $H(\omega)$ instead of $H(z)$.
The passband and stopband are respectively $[0,\omega_p]$ and $[\omega_s,\pi]$, where $\omega_p, \omega_s$ are given. The stopband energy of the FIR filter is
$$E_s = \frac{1}{\pi}\int_{\omega_s}^{\pi}|h(\omega)|^2d\omega.$$
To design such a linear-phase filter, we minimize the stopband energy under modulus constraints involving two parameters $\gamma_p,\gamma_s$:
\begin{align*}
\min\limits_{ H\in \Hi[z]} & \ E_s \\
  \text{s.t.} \ \ \ & \ |H(\omega)-1|\leq \gamma_p, \ \forall \omega\in[0,\omega_p],\\
  & \ | H(\omega)|\leq \gamma_s, \qquad \forall \omega\in[\omega_s,\pi] \,.
\end{align*}
As shown in \cite[\S~5.1.1]{dumi2017}, the above problem can be reformulated as:
\begin{equation}\label{eq:sdp}
    \begin{array}{cl}
    \min\limits_{h,\, Q_1,\dots,\, Q_7} & \ h^T \tilde{C} h \\
& (1+\gamma_p) 1_{k=0}   - {h}_k = L_k(Q1) \,, \\
    & {h}_k - (1-\gamma_p)1_{k=0}  = L_{k,0,\omega_p}(Q_2,Q_3) \,,\\
     \text{s.t.} & \gamma_s 1_{k=0}   - {h}_k  = L_{k,\omega_s,\pi}(Q_4,Q_5) \,,\\
    & \gamma_s 1_{k=0}   + \, {h}_k = L_{k,0,\omega_p}(Q_6,Q_7) \,, \quad k = 0,\dots,d \,, \\
    &  Q_1 \succeq 0,\dots,Q_7 \succeq 0, \\
\end{array}
\end{equation}
where $Q_1,Q_2,Q_4,Q_6$ are real $(d+1)\times(d+1)$-matrices, $Q_3,Q_5,Q_7$ are real $(d-1)\times(d-1)$-matrices; $L_k(A):=\trace (\Theta_k A)$ and
\begin{equation}\label{eq:Lk}
    \begin{array}{ll}
        L_{k,\alpha,\beta}(A,B): = & \trace (\Theta_k A) + \trace \Big( \big(\frac{a+b}{2} (\Phi_{k-1} + \Phi_{k+1})
        \\
        & -(ab+\frac{1}{2}) \Phi_k - \frac{1}{4}(\Phi_{k-2} + \Phi_{k+2})\big)B\Big),
    \end{array}
\end{equation}
with $a = \cos \alpha, b=\cos \beta$; $\Theta_k\in\R^{(d+1)\times(d+1)}, \Phi_k\in\R^{(d-1)\times(d-1)}$ are the elementary Toeplitz matrices with ones on the $k$-th diagonal and zeros elsewhere (they are zero matrices whenever $k$ is out of range); $C={\rm Toep}(c_0,\dots,c_d)$ is the Toeplitz matrix with the first row $(c_0,\dots,c_d)$, where
$$c_k= \begin{cases}
1-\frac{\omega_s}{\pi}, & \text{ if } \  k=0,\\
-\frac{\sin k\omega_s}{k\pi}, & \text{ if } \ k > 0,
\end{cases}
$$

$$\tilde{C} = P^TCP\succeq 0, \
P=\begin{bmatrix}
  0  & J_d \\
   1 & 0 \\
   0 & I_d
\end{bmatrix} \,,$$
where $J_d$ and  $I_d$ denote the counter identity and identity matrices of size $d$, respectively.
By contrast with the unconstrained case (Algorithm $\csosthree$), this program involves 7 Gram real matrix
variables and $d+1$ real variables $h_0, \dots,h_d$, which are the coefficients of the polynomial corresponding to the filter.

After solving \eqref{eq:sdp}, we obtain numerical values for the coefficients of ${h}$ and the entries of $Q_1,\dots,Q_7$, which are further rounded to $\hat{h}$ and $\hat Q_1,\dots, \hat Q_7$.
To project $\hat Q_1$ to a matrix $Q_1$ satisfying the first set of equality constraints in SDP \eqref{eq:sdp}, we apply the formula in Line \ref{line:projPP} of Algorithm $\csosthree$ after replacing $f_k$ by  $p_k:=(1+\gamma_p)1_{k=0} - \hat{h}_k$.
Similarly, we obtain matrices $Q_2$ and $Q_3 := \hat{Q_3}$ satisfying the second set of equality constraints in SDP \eqref{eq:sdp}, after substitution by
\if
Solving \eqref{eq:sdp} yields numerical approximations $\tilde {h}$, $\tilde Q_1,\dots,\tilde Q_7$, which are further rounded to $\hat{h}$, $\hat Q_1,\dots, \hat Q_7$.
To project $\hat Q_1$ to a matrix $Q_1$ satisfying exactly the first set of equality constraints in \eqref{eq:sdp}, we apply the formula in Line \ref{line:projPP} of Algorithm $\csosthree$ after replacing $f_k$ by  $p_k:=(1+\gamma_p)1_{k=0} - \hat{h}_k$.
Similarly, we obtain matrices $Q_2$ and $Q_3 := \hat{Q_3}$ satisfying exactly the second set of equality constraints in SDP \eqref{eq:sdp}, after substitution by
\fi
\begin{equation*}
    \begin{array}{ll}
& \hat{h}_k - (1-\gamma_p)1_{k=0} - \trace \Big( \big(\frac{a+b}{2} (\Phi_{k-1} + \Phi_{k+1})
-(ab+\frac{1}{2}) \Phi_k \\
        & - \frac{1}{4}(\Phi_{k-2} + \Phi_{k+2})\big){Q}_3\Big),
    \end{array}
\end{equation*}
Eventually, similar projection steps provide the  remaining matrices $Q_4, \dots, Q_7$ so that all equality constraints in  \eqref{eq:sdp} hold exactly.


As in \cite[Example 5.1]{dumi2017}, we design a  filter with parameters $d=25$, $\omega_p=\pi/5$, $\omega_s=\pi/4$, $\gamma_p=1/10$ (corresponding
to a passband ripple of 1.74 dB), and $\gamma_s=0.0158$ (a stopband attenuation of 36 dB).
We first obtain a numerical lower bound of the stopband energy $E'_s = 4.461501\times~10^{-5}$.
However, this bound happens to be inexact as the Gram matrices obtained after the projection step are not positive semidefinite anymore.
To overcome this certification issue, we replace the last constraint in \eqref{eq:sdp} by $Q_7 - 10^{-9}I_{24} \succeq 0$.
Doing so, we can successfully project the approximate Gram matrices into exact ones with positive eigenvalues, and obtain a certified exact lower bound  of $E_s = 4.461503\times 10^{-5}$ in $0.74$ seconds.
\vspace{-0.3cm}
\section{Conclusion and discussion}
\label{sec:conclusion}
We have designed three algorithms of polynomial bit complexity to compute
weighted sums of Hermitian squares decompositions for trigonometric univariate
polynomials positive on the unit circle with Gaussian coefficients. 
%
Note that positivity of such a trigonometric polynomial $f$ is equivalent to that of a polynomial $a_0 + \sum_{k=1}^d a_k \cos(k t) + b_k \sin(k t)$ for all $t \in [0,2 \pi]$, where $a_k,b_k$ are rational coefficients obtained from the coefficients $f_i$.
In turn, if we do the change of variables $t=2 \arctan(x)$, the trigonometric polynomial becomes a rational function, whose denominator is a power of $(1+x^2)$. 
Thus this boils down to proving positivity of a real univariate polynomial, and one can apply the methods from \cite{univsos}.
We also leave the situation where the input polynomial vanishes on the unit
circle for future investigation.

\noindent
{\em Acknowledgments.} We thank the anonymous referees for their remarks and
suggestions.    
This work has been supported by European Union's Horizon 2020 research and
    innovation programme under the Marie Sk\l{}odowska--Curie Actions, grant
    agreement 813211 (POEMA). \\
\vspace{-0.3cm}


\end{document}